\documentclass[aps,prl,preprint,superscriptaddress]{revtex4}

\usepackage{amssymb,amstext,amsthm,amsmath}
\usepackage{multirow} 
\usepackage{changepage} 
\usepackage{graphicx}
\usepackage{xcolor}

\newtheorem{theorem}{Theorem}

\newcommand{\rep}[1]{{\mathbf{#1}}}
\newcommand{\repbar}[1]{{\overline{\mathbf{#1}}}}

\usepackage{hyperref}
\usepackage{cleveref}
\usepackage{blindtext}
\usepackage{graphicx}
\usepackage{tikz-cd}

\begin{document}


\title{A Survey of Family Unification Models with Bifundamental Matter}



\author{Elijah Sheridan}
\email{es888@cornell.edu}
\affiliation{Department of Physics and Astronomy, Vanderbilt University, Nashville, TN 37235}
\affiliation{Department of Physics, Cornell University, Ithaca, NY 14853, USA}

\author{Thomas W. Kephart}
\email{tom.kephart@gmail.com}
\affiliation{Department of Physics and Astronomy, Vanderbilt University, Nashville, TN 37235}

\date{\today}

\begin{abstract}

Extensions of the Standard Model have been attempted from the bottom up and from the top down yet there remains a largely unexplored middle ground. In this paper, using the Mathematica package LieART, we exhaustively enumerate embeddings of the Standard Model within the class of theories with  bifundamental fermions in product gauge group $SU(a) \times SU(b) \times SU(c)$ with no more generators than $E_6$, while achieving SM family unification rather than replication. We incorporate simple phenomenological constrains and find $151$ unique models, including $9$ that have only vector-like particle content beyond-the-Standard Model (BSM) particles, which we conjecture belong to $5$ infinite families of such models. We describe the potentially most viable models: namely, the $9$ with strictly vector-like BSM content along with the $29$ models we found with no more than $30$ additional BSM chiral particles. These include models with fractional electric charge color singlets, and hence magnetic monopoles with multiple Dirac charge. This latter collection of models predicts chiral particles with masses near the electroweak scale  accessible to current and future collider experiments.
\end{abstract}


\maketitle


\section{Introduction}

It has been a long held and very ambitious goal to derive the standard model (SM) of 
particle physics from string theory. While this has yet to be accomplished, one can still make
progress using string theory as a guide. On the string side various compactafications from ten dimensions (10D)
down to our physical four dimensions (4D) result in a variety of promising gauge theories. These top down models 
have intriguing properties,  but are still some distance from the SM. From the SM side one can try a bottoms up
approach and attempt to unify the SM into a larger gauge group to generate something close to what we obtain
from strings. 

There is also a middle ground where we allow more freedom in our choice of gauge group and
particle spectrum from what has been derived from strings at present, and a wider range of unifications 
than typically used  for the SM. These ``string inspired'' models hold promise for matching high and low scales and are the topic of this work. In particular, we investigate a class of bifundamental models inspired by orbifolding $AdS_5 \times S^5$ by a discrete group $\Gamma$ \cite{Lawrence:1998ja}, which generates 4D theories with gauge groups of the 
type $SU(n_1N)^{q_1}\times SU(n_2N)^{q_2}\times ...$ where the $n_i$s are the dimensions of the irreducible representations (irreps) of $\Gamma$, and the $q_i$ also depend on the choice of $\Gamma$. The fermions will be in bifundamental irreps of the gauge group. (For details see \cite{Lawrence:1998ja}.)  These models are often called quiver gauge theories \cite{Uranga:2000ck,Cachazo:2001gh,Brax:2002rz,Antebi:2005hr,Burrington:2006uu,Burrington:2007mj,Frampton:2007fr,Nekrasov:2012xe,He:2018gvd}.

Here we investigate theories in that middle ground, with a focus on smaller gauge groups   of the form $SU(a)\times SU(b)\times SU(c)$. That is, we limit our study to models whose gauge groups have three $SU$ factors. Such theories could arise, for example, from orbifold constructions with small discrete group $\Gamma$ or a larger discrete group followed by rounds of spontaneous symmetry breaking (SSB).

Beyond the well studied cases of  $SU(4)\times SU(2)\times SU(2)$, the Pati-Salam model \cite{Pati:1974yy} and $SU(3)\times SU(3)\times SU(3)$ trinification (TR) \cite{trin}, the cases $SU(3)\times SU(3)\times SU(4)$, $SU(3)\times SU(4)\times SU(4)$ and
$SU(3)\times SU(3)\times SU(5)$, are the smallest 
gauge groups of interest. It is at this stage (i.e., at the ``$abc$'' scale) that we set up our initial models, and require the product group have dimension no greater than that of $E_6$ and that the initial fermions  live only in bifundamental irreps of the gauge groups such that there are no gauge anomalies.  We refer to such models as Small Bifundamental Theories (SBTs). Further SSB then takes a SBT to the SM gauge group. The reason we study these cases,  beyond their relevance for matching low  and high energy scales of string theoretic models and the SM, respectively, is that they can have natural family unification, where SM fermion families must come in multiples of three to be chiral anomaly free in the initial $abc$ gauge group.

Starting with the minimal set of $SU(a)\times SU(b)\times SU(c)$ bifundamentals, we look at all possible models that can contain the SM with only 3 families. Even this minimal choice generates a large number of models, so to deal with them we use the Mathematica package LieART \cite{Feger:2019tvk,Feger:2012bs} to generate and analyze them. As we will find, the number of models grows quickly with the size of the initial product gauge group. To improve the utility of our results, we incorporate some basic phenomenological information. If some of the extra fermions in a chiral extension model are charged and remain massless (i.e., with no allowed mass term) after SSB at the SM level, then we discard the model as there are no observed massless charged particles. Additionally, theories with BSM chiral particles that are dual to SM particles are also omitted, as we require, and it has been empirically verified that none of the SM particles belong to vector-like pairs.
We also note that models with four or more families are currently disfavored but we will record them for completeness. 

To summarize, the three family models fall into subclasses which are: (i) ``pristine models'' where only the three families are chiral and all other fermions are vector-like (VL) (i.e., come in left-right pairs) and hence get masses at the $abc$ breaking scale, and (ii) ``chiral extensions'' where there are three families plus some extra chiral fermions, where the extra fermions in these models can be leptionic, hadronic or both.

The  pristine models that correspond directly to the SM at low energy with new physics only due to particles and interactions at the $abc$ scale. While the $abc$ scale is model dependent, it could be relatively low if the initial gauge couplings of 
$SU(a)$,  $SU(b)$ and $SU(c)$ differ, which is allowed depending on how the $abc$ gauge group is obtained from an initial quiver gauge theory. Otherwise it is likely to be at a typical GUT scale \cite{Georgi:1974sy}. 

More interesting but also more dangerous, are the chiral extension models, which lead to new physics guaranteed to be near the electro-weak (EW) scale.
Assuming we have rejected the models with massless charged particles, naturalness and perturbative stability requires the remaining chiral extension models to have chiral fermion near the EW scale. Current accelerator constraints can  then be applied to place further restrictions on these models. Hence, we anticipate these models will be either tightly constrained or eliminated by current data.

Our approach to finding potentially viable models is as simple as possible.
I.e., our interest here is limited to the fermionic spectrum of these models. Hence, we use the minimal phenomenological constrains discussed above plus generic
group theory analysis to identify these models. We do not look at the detailed 
phenomenology of any of these models, nor do we compare with or cite the many  interesting quiver gauge theory derived models or their phenomenology. While we hope to return elsewhere to the phenomenology of the most interesting models found here, such investigations are beyond the scope of this present work. This also means that the spontaneous symmetry breaking (SSB) paths we take are only chosen for their expediency for arriving at the fermion spectrum of the models, not to achieve a realistic phenomenology. Again, the detailed phenomenology of individual models is left to later work.
 
 We begin with a discussion of the set of product groups to be investigated and the method we will use, after which we tabulate our results and present our conclusions.
 
\section{Search Method}

\subsection{Product Groups}

It is well known that the gauge group of the SM can be most simply embedded in the chain of grand unified theories (GUTs),
$SU(5)\subset SO(10)\subset E_6$, but other possibilities abound, including $SU(N)$ for $N > 5$ and various product groups, where the minimal choices are the Pati-Salam group $SU(2)\times SU(2)\times SU(4)$ \cite{Pati:1974yy} which can be embedded in $SO(10)$ and the trinification group $SU(3)\times SU(3)\times SU(3)$ \cite{trin} which embeds in $E_6$ \cite{Gursey:1976ki,Achiman}. There are also models based of the product of more than three $SU(N)$ factors like the quartification models \cite{Foot,Babu:2003nw,Chen:2004jza,Demaria:2005gk,Demaria:2006uu,Demaria:2006bd,Eby:2011ph,Dent:2020jod}.

Here we will be interested in a certain class of product models that can be inspired by but not necessarily derived from string theory. Specifically, we consider models similar to those from orbifolded compactifications $AdS_5 \times S^5/\Gamma$ where $\Gamma$ is a finite group, either abelian \cite{Kachru:1998ys} or nonabelian \cite{Frampton:2000zy,Frampton:2000mq,Frampton:2007fr}. In the example of embedding the SM in a $SU(3)\times SU(3)\times SU(4)$ (334) model \cite{Kephart:2001ix,Kephart:2006zd}, several possible inequivalent embeddings lead to PS and TR type models, both with additional phenomenology plus new, and previously unexplored models. Most of these models had fractionally charged color singlets \cite{Kim:1980yk,Goldberg:1981jt}, and hence multiply Dirac charged magnetic monopoles. (Discussions of recent and ongoing searches for magnetic monopoles and particles with nonSM charges can be found in \cite{Alimena:2019zri,MoEDAL:2019ort,Mavromatos:2020gwk,Song:2021vpo}.) These studies were carried out by hand, and hence limited in scope, but we can now employ LieART to handle the more tedious part of the search for viable models and to recheck the previous results. This will allow us to extend our analysis to include a complete classification of the more general cases. 
There are many product models where the gauge group can be smaller, i.e., have as many or fewer generators, than $E_6$.
I.e., they include examples mentioned above and $SU(a)\times SU(b)\times SU(c)$ with $(a,b,c)$ such that
$a^2+b^2+c^2-3 \leq 78$.
As stated above, we refer to these models (with bifundamental fermionic content) as Small Bifundamental Theories (SBTs).
We will analyze all these cases and hence provide an overall prospective on the general class of string inspired $abc$ product group models.  We note that the $E_6\rightarrow TR \rightarrow SM$ and $SO(10) \rightarrow  PS\rightarrow SM$ decompositions have been extensively explored, while the $SU(a)\times SU(b)\times SU(c)\rightarrow (TR$ or $PS)$ $\rightarrow SM$ pathways have been neglected by comparison. See however \cite{Kephart:2017esj,Raut:2022ryj}.

One of the more pervasive new features we find is that many of these product group models contain chiral fermions  types not present in the SM. Hence they are potential candidates for new light ($\sim$ TeV) fermions near the electroweak (EW) scale. E.g., fractionally charged leptons could provide very interesting  BSM physics.  

To set the stage, we need to review the SM particle content and its embedding in some familar models. A standard family $\rep{F}$  contains fermions in the following  irreps of $SU(3)_C\times SU(2)_L\times U(1)_Y$
\begin{equation}
    \rep{F} = (\rep{3},\rep{2})_{\frac{1}{6}} + 
    (\rep{\bar{3}},\rep{1})_{\frac{1}{3}} + (\rep{\bar{3}},\rep{1})_{-\frac{2}{3}} + (\rep{1},\rep{2})_{-\frac{1}{2}} + (\rep{1},\rep{1})_{1}
\end{equation}
where we may or may not want to include a right-handed neutrino $(1,1)_{0}$. The family embeds directly in $SU(5)$ as 
$\rep{F} = \rep{{\bar 5}} + \rep{10}$ and in a spinor of $SO(10)$ if we include the right-handed neutrino $\rep{F} + \rep{1} = \rep{16}$. Including another neutral singlet and a conjugate $SU(5)$ vector pair we can write $\rep{F} + \rep{5} + \repbar{5} + \rep{1} + \rep{1} = \rep{27}$ as an $E_6$ family. The SM family also fits in the TR model where the trinification group is $SU(3)\times SU(3)\times SU(3)\subset E_6$. Here
\begin{equation}
    \rep{27} \rightarrow (\rep{3},\rep{\bar{3}},\rep{1}) + (\rep{1},\rep{3},\rep{\bar{3}}) + (\rep{\bar{3}},\rep{1},\rep{3}),
\end{equation}
where the bifundamental matter contains $\rep{F}$ plus additional vector like particles after SSB to the SM.
 
On the other hand  the Pati-Salam group is $SU(2)\times SU(2)\times SU(4)\subset SO(10)$ and
we expect bifundamental matter of the form 
\begin{equation}
    (\rep{4},\rep{{\bar 2}},\rep{1})+(\rep{{\bar 4}},\rep{1},\rep{2})+(\rep{1},\rep{2},\rep{\bar{2}}).
\end{equation}
However, the $\rep{2}$ is pseudo real so we can replace the $\rep{2}$ bars with $\rep{\bar{2}}$s and note that $(\rep{1},\rep{2},\rep{2})$ is vectorlike and has a direct mass term, hence it can be dropped, leaving us with just the right fermions to fit into the spinor of $SO(10)$
\begin{equation}
    \rep{16} \rightarrow  (\rep{4},\rep{2},\rep{1})+({\rep{\bar 4}},\rep{1},\rep{2}).
\end{equation}
 
Before looking at new individual models, let us first describe the scheme by which we find such models.
 
\subsection{General approach}
 
LieART is a Mathematica application to deal with Lie algebra representations and their tensor products. It allows us to implement an exhaustive search for all SBTs.

Indeed, we consider
\begin{equation}
    \mathcal{G} = \{G = SU(a) \times SU(b) \times SU(c) \; | \; \dim G \leq 78 = \dim E_6; a \geq 3; b, c \geq 2\}.
\end{equation}
We identify permutations: e.g., $SU(4) \times SU(3) \times SU(3) \sim SU(3) \times SU(4) \times SU(3)$. As required, all elements of $\mathcal{G}$ admit $G_{SM} = SU(3)_C \times SU(2)_L \times U(1)_Y$ as a subgroup and satisfy the desired low dimensionality. The set $\mathcal{G}$ contains $39$ groups, from $SU(3) \times SU(2) \times SU(2)$ to $SU(7) \times SU(4) \times SU(4)$.
 
For $G = SU(a) \times SU(b) \times SU(c) \in \mathcal{G}$, the goal is to take anomaly free bifundamental representations $\rep{R}$ of $G$ and decompose them to representations $\rep{R}'$ of $G_{SM}$ such that $\rep{R}_{SM} \subset \rep{R}'$. Here, we recall that
\begin{align}
    \rep{R} \label{eq:bifundamental}
    &= \frac{c}{n}(\rep{a}, \repbar{b}, \rep{1}) + \frac{a}{n}(\rep{1}, \rep{b}, \repbar{c}) + \frac{b}{n}(\repbar{a}, \rep{1}, \rep{c}) 
\end{align}
where $n = \text{gcd}(a,b,c)$ and
\begin{align}
    \rep{R}_{SM} \label{eq:SMrep}
    &= 3(\rep{F} + \rep{1}) = 3[ (\rep{3}, \rep{2})_\frac{1}{6} + (\repbar{3}, \rep{1})_\frac{1}{3} 
    + (\repbar{3}, \rep{1})_{-\frac{2}{3}} + (\rep{1}, \rep{2})_{-\frac{1}{2}} + (\rep{1}, \rep{1})_1 + (\rep{1}, \rep{1})_0].
\end{align}
Clearly a $G = SU(a) \times SU(b) \times SU(c) \in \mathcal{G}$ gives rise to a minimal anomaly free  bifundamental set $\rep{R}$ with $N_P(G) = 3\frac{abc}{n}$ particles.

We understand the symmetry breaking process by separating it into two parts: first, to $G$ we can associate $\tilde{G} = SU(3)_C \times SU(2)_L \times U(1)^m$, achieved by iteratively applying the symmetry breaking $SU(n) \to SU(n-1) \times U(1)$ to each factor in $G$, i.e., we consider only regular embeddings. In particular, $m = (a - 3) + (b - 2) + (c - 1)$. We refer to this symmetry breaking as non-Abelian breaking (NAB). (We do not give specifics other than noting that this step of spontaneous symmetry breaking (SSB) can be accomplished with adjoint scalars.) Then, we break $\tilde{G} \to G_{SM}$ (i.e., $U(1)^m \to U(1)_Y$) by choosing a linear combination of the $m$ $U(1)$ factors to identify with $U(1)_Y \subset G_{SM}$. We denote this symmetry breaking by Abelian breaking (AB). SSB is now more involved and requires that we break rank by $m-1$ while leaving the linear combination $Y$ unbroken. The details are not needed at present and outside the scope of what we study here. 

For the remainder of this section, we give an overview of how we performed exhaustive searches over all possible NABs and ABs  (because for each $G \in \mathcal{G}$ there can be several inequivalent paths to $G_{SM}$) and briefly summarize our results. The next section will  explicitly describe the technical details of our algorithm for computing ABs, while the current section  features a more qualitative overview.

\subsection{Non-Abelian Symmetry Breakings}

A NAB is specified by a choice of $SU$ factor of $G$ in which to embed $SU(3)_C \subset G_{SM}$ and a choice of a distinct $SU$ factor in which to embed $SU(2)_L$. Moreover, we are interested only in the NABs for which each $SU(3)_C \times SU(2)_L$ irrep in $\rep{R}_{SM}$ can be found in the decomposition of $\rep{R}$ to $\tilde{G}$. While, \textit{a priori}, this property depends upon the NAB and not merely on the gauge group, we find that either all viable NABs for a given $G \in \mathcal{G}$ have this property or none of them do. In particular, let $\mathcal{G}' \subset \mathcal{G}$ denote the set of groups with this property: it turns out that
\begin{equation}
    \mathcal{G}' = \{ ( SU(a) \times SU(b) \times SU(c) \in \mathcal{G} \; | \; 2 \notin (a,b,c) \; \& \, \gcd(a,b,c) = 1 \}
\end{equation}
To understand why $\mathcal{G}'$ has the structure that it does, observe that for $SU(N) \times SU(2) \times SU(2)$, $SU(3)_C$ must be contained in the $SU(N)$ and $SU(2)_L$ must be in an $SU(2)$, in which case all the leptons have the same hypercharge, which is incompatible with the SM. With a bit more work one can show that $3$ SM families can not be obtained from $SU(N) \times SU(N') \times SU(2)$. Additionally, we find that the greatest common divisor must also be trivial in our case because otherwise the coefficients Eq. (\ref{eq:bifundamental}) will be less than $3$ and prevent decomposition into three SM families. Hence we find the set $\mathcal{G}'$ contains only 13 groups. From here, to parameterize our NABs, we define
\begin{equation}
    \mathcal{E} = \{ (a,b,c) \; | \; SU(a) \times SU(b) \times SU(c) \in \mathcal{G}', a,b,c \geq 3 \} 
\end{equation}
where each element of $\mathcal{E}$ is an ordered triple. We interpret $E = (a,b,c) \in \mathcal{E}$ as the choice of $SU(a) \times SU(b) \times SU(c) \in \mathcal{G}'$ as the gauge group; $SU(a) \to SU(3)_C, SU(b) \to SU(2)_L$ as the NAB; and $\rep{R}$, as given by Eq. (\ref{eq:bifundamental}), as the fermion representation of the unbroken theory, which decomposes to $\rep{R}_E$ per the above symmetry breaking. The set $\mathcal{E}$ contains $54$ elements. We note that groups in $\mathcal{G}$ contribute $1$, $3$, or $6$ triples to $\mathcal{E}$ depending on whether $a,b,c$ are identical, contain an identical pair, or are distinct, respectively. For example, $SU(4) \times SU(3) \times SU(3) \in \mathcal{G}'$ makes three contributions to $\mathcal{E}$, $(4,3,3), (3,4,3)$ and $(3,3,4)$.

We need not consider regular embeddings wherein $SU(3)_C \times SU(2)_L \subset G_{SM}$ is contained in a single $SU$ factor in $G$. This is because, when breaking $SU(n)$, the fundamental irrep $\rep{n}$ can never decompose to the $(\rep{3},\rep{2})$ of $\rep{R}_{SM}$, hence such SM embeddings are precluded for bifundamental representations.
 
\subsection{Abelian Symmetry Breakings}

Given a NAB $E = (a,b,c) \in \mathcal{E}$ with decomposed irrep $\rep{R}_E$ for $\tilde{G} = SU(3)_C \times SU(2)_L \times U(1)^m$, selecting an AB requires a choice of $m$ rational numbers $x = (x_1, \dots, x_n)$: then we have $\tilde{G} \to G_{SM}$, or $U(1)^m \to U(1)_Y$, by letting the hypercharge operator $Y$ be given by
\begin{equation}
    \label{eq:U1Y_from_U1m}
    Y = \sum_{i=1}^m x_i Z_i
\end{equation}
where $Z_i$ denotes the charge operator for the $i$th $U(1)$ factor in $\tilde{G}$. In summary, our symmetry breaking and representation decomposition takes the schematic form
\begin{center}
    \begin{tikzcd}[row sep=huge]
        G = SU(a) \times SU(b) \times SU(c) \in \mathcal{G}' \arrow[d, "{\text{NAB } (E \, \in \, \mathcal{E})}"'] & \rep{R} \arrow[d] \\ 
        \tilde{G} = SU(3)_C \times SU(2)_L \times U(1)^m \arrow[d, "{\text{AB } (x \, \in \, \mathbb{R}^m)}"'] & \rep{R}_E \arrow[d] \\
        G_{SM} & \rep{R}_x
    \end{tikzcd}
\end{center}

The choice of $x = (x_1, \dots, x_n)$ is necessary but insufficient for determining an AB which gives rise to the SM particles, but we will elaborate on the remaining required information in the next section. This next section will describe an algorithm for exhaustively searching for all $x \in \mathbb{R}^m$ (where $m$ changes between elements of $\mathcal{E}$, of course) which yield ABs such that $\rep{R}_x \supset \rep{R}_{SM}$, where $\rep{R}_x$ is the decomposition of $\rep{R}_E$ per the prescribed AB. 

Finally, to each of these ABs we will ultimately apply two readily-accessible phenomenological constraints. First, we accept only ABs that yield representations $\rep{R}_x$ for $G_{SM}$ which contain no massless charged chiral fermions after electroweak symmetry breaking, since they are strictly phenomenologically forbidden. Second, we do not allow any of the extra fermions in the chiral extension models to have quantum numbers conjugate to any SM particles. This would allow them to pair up with SM particles and produce mass terms well above the electroweak scale, and hence the model would not contain the requisite complete three families at low energy. After imposing these constraints, we are left with a set of (possibly) physically viable models.

\section{Search Algorithm}

As we have seen, the nonAbelian symmetry breaking is straight forward since there are only a few pathways to consider. On the other hand, the Abelian symmetry breaking is considerably more complicated and for this we have developed the algorithm we now describe.

\subsection{Abelian Breaking Algorithm}

We begin with a representation $\rep{R}_E$ of $\tilde{G} = SU(3) \times SU(2) \times U(1)^m$ resulting from the decomposition of a bifundamental $\rep{R}$ for $G = SU(a) \times SU(b) \times SU(c) \in \mathcal{G}'$ undergoing a NAB specified by $E = (a, b, c) \in \mathcal{E}$. In general, $\rep{R}_E$ will contain terms of the following form
\begin{equation}
    \label{eq:rbar}
    \begin{aligned}
        \rep{R}_E &=
    \sum_{\alpha = 1}^A N_{\alpha} (\rep{3}, \rep{2})_{q_{\alpha 1} \dots q_{\alpha m}}
    + \sum_{\beta = A+1}^{A+B} N_{\beta}(\repbar{3}, \rep{1})_{q_{\beta 1} \dots q_{\beta m}} \\
    &\qquad + \sum_{\gamma =A+B+1}^{A+B+C} N_{\gamma} (\rep{1}, \rep{2})_{q_{\gamma 1} \dots q_{\gamma m}}
    + \sum_{\delta =A+B+C+1}^{A+B+C+D} N_{\delta } (\rep{1}, \rep{1})_{q_{\delta 1} \dots q_{\delta m}}.
    \end{aligned}
\end{equation}
Here, the $q_{i1}, \dots, q_{im}$ are the $m$ integer $U(1)$ charges associated with the $i^{th}$ irrep term in this sum and the integer $N_i$ denotes the multiplicity of that term, where $i = \alpha, \beta, \dots$. Our task here is to reduce $\tilde{G}$ to $G_{SM}$, decomposing $\rep{R}_E$ to $\rep{R}_x$ such that $\rep{R}_{SM} \subset \rep{R}_x$. We recall that this ultimately requires the determination of $m$ real numbers $x = (x_1, \dots, x_m)$ from which $U(1)_Y \subset G_{SM}$ is formed from $U(1)^m \subset G$. We can solve for this $x$ by selecting terms from Eq. (\ref{eq:rbar}) to identify with the SM particles. Explicitly, this means choosing a $\kappa \in \{1, \dots, A\}$, some $\lambda, \mu \in \{A+1, \dots, A+B\}$ ($\lambda \neq \mu$), a $\nu \in \{A+B+1, \dots, A+B+C\}$, and some $\xi, \rho \in \{A+B+C+1, \dots, A+B+C+D\}$ ($\xi \neq \rho$). This is the additional information needed to determine a AB as referenced in the previous section: an AB is determined by a six-tuple $w = (\kappa, \lambda, \mu, \nu, \xi, \rho)$ of integers providing identifications of terms in $\rep{R}_E$ with $\rep{R}_{SM}$ and an $m$-tuple $x = (x_1, \dots, x_m)$ of rational numbers determining a hypercharge which ensures the six distinguished particle irreps have the correct SM charges.

Given a six-tuple $w$ choosing which particles are to constitute the SM, we can solve for the $x$ determining the hypercharge. Our algorithm works by exhaustively iterating over all choices of $w$ and solving for $x$ in each case. From the combinatorics, the number of possible viable ABs is bounded above by $A \cdot B \cdot (B-1) \cdot C \cdot D \cdot (D - 1)$.

Concretely, given a choice of $w = (\kappa, \lambda, \mu, \nu, \xi, \rho)$, $x = (x_1, \dots, x_m)$ is given by the following linear equation.
\begin{equation}
    \begin{pmatrix}
    q_{\kappa 1} & q_{\kappa 2} & \dots & q_{\kappa m} \\ 
    q_{\lambda 1} & q_{\lambda 2} & \dots & q_{\lambda m} \\ 
    q_{\mu 1} & q_{\mu 2} & \dots & q_{\mu m} \\ 
    q_{\nu 1} & q_{\nu 2} & \dots & q_{\nu m} \\ 
    q_{\xi 1} & q_{\xi 2} & \dots & q_{\xi m} \\ 
    q_{\rho 1} & q_{\rho 2} & \dots & q_{\rho m}
    \end{pmatrix}
    \begin{pmatrix}x_1 \\ x_2 \\ x_3 \\ x_4 \\ \vdots \\ x_m \end{pmatrix} =
    \begin{pmatrix}\tfrac{1}{6} \\[0.1em] \tfrac{1}{3} \\[0.1em] -\tfrac{2}{3} \\[0.1em] -\tfrac{1}{2} \\[0.1em] 1 \\ 0\end{pmatrix}
    \label{linearEQ}
\end{equation}
As established above, the hypercharge operator $Y$ for $U(1)_Y$ is formed from the $x_i$ per Eq. (\ref{eq:U1Y_from_U1m}). 

As evidenced by the structure of our algorithm, we do not consider potential embeddings where we identify, say, $3(\rep{3}, \rep{2})_\frac{1}{6} \in \rep{R_{\text{SM}}}$ with only two distinct terms in $\rep{R'_{abc}}$, one with $N_i = 1$ and the other with $N_i = 2$. Thus, we enforce $N_i \geq 3$ for $i \in w = (\kappa, \lambda, \mu, \nu, \xi, \varepsilon)$ to ensure we achieve at least the three SM families. In particular, this is why the upper bound presented earlier failed to be an equality: in general, we are excluding terms with $N_i < 3$. 

Observe, of course, that the linear equation Eq.(\ref{linearEQ}) need not be solvable, as it may be overdetermined or underdetermined depending upon the sign of $m - 6$. We keep only cases where this system has a solution. In practice, over all elements of $\mathcal{E}$, we found that about $80\%$ of allowed choices of $w$ yield solvable linear equations. Certainly solutions also need not be unique: indeed, the space of solutions is an affine subspace with dimension equal to that of the null space of the matrix in Eq.(\ref{linearEQ}). In particular, if $m - 6$ is positive, this null space is necessarily non-trivial and there are infinitely many solutions. Each distinct (rational) solution in this affine space will, in general, yield different hypercharges for all particles not being identified with SM particles (that is, all particles with indices not equal any of the components of $w$). We handle this ambiguity by picking out the solution $x$ with the smallest norm: that is, the element of the solution space whose Euclidean distance from the origin is least. (This to be the default behavior of Mathematica's \texttt{LinearSolve} function\footnote{In particular, though the documentation doesn't make this obvious, the authors believe \texttt{LinearSolve} employs the Moore-Penrose inverse to solve underdetermined systems, which would indeed yield the least norm solution}). Physically, this choice would be made by giving VEVs to scalar fields that break the $m - 6$ $U(1)$ generators in such a way as to fix the minimal norm solution.
 
The minimal norm choice is physically motivated.  
Since the charges of the SM particles are already fixed, our choice of $x$ from a non-trivial solution space for eq.(\ref{linearEQ}) for only affects the non SM particles, some of which may be chiral. By choosing the smallest norm in $x_i$ we minimize those charges, which is a preferable choice as can be argued via an appeal to  Occam's razor. Moreover, we recall that we require rational solutions (to yield rational electric charges and preserve charge quantization), and it is a mathematical result that because the coefficients of eq.(\ref{linearEQ}) are rational, the least norm solution is also strictly rational as we desire. This is not a difficult result, but for clarity and thoroughness we prove it briefly as Theorem \ref{th} in the Appendix. In summary, then, the least norm solution is physically motivated in the sense that the total charge of the bifundamental irreps is minimized and because no irrational charges are observe in nature. Nevertheless, we wish to emphasize that when a choice of $w = (\kappa, \lambda, \mu, \nu, \xi, \rho)$ can yield a matrix of coefficients with a non-trivial kernel, which would necessarily entail the existence of infinitely many distinct models, of which we are only considering one in this paper, the minimal choice.

We have encountered a considerable number of instances where distinct choices of $w$ give rise to the same $x$ (and thus, the same hypercharge) and instances where distinct hypercharges give rise to the same particle content. Finally,   to avoid models of less interest, we are only treating three families as SM content: the appearance of further SM families within models are treated as BSM content.

\section{Results}

Let us denote the set of representations $\rep{R}_x$ of $G_{SM}$ arising from the ABs previously described (i.e., each containing $\rep{R}_{SM}$) by $\mathcal{M}$, and in particular define $\mathcal{M}_E$ to be the representations coming from the NAB given by $E \in \mathcal{E}$ (i.e., $\mathcal{M}$ is the union of the $\mathcal{M}_E$). Note that $\mathcal{M}$ is smaller than the number of surviving ABs, as we have as yet made no effort to remove duplicate representations, and different NABs and ABs that can (and often do) give rise to the same representations. Each element $\rep{R}_x \in \mathcal{M}$ takes the form $\rep{R}_x = \rep{R}_{SM} + \rep{R}_{VL} + \rep{R}_C$ where $\rep{R}_{VL}$ and $\rep{R}_C$ are vector-like and chiral representations respectively. From a physical standpoint, we are especially interested in elements of $\mathcal{M}$ satisfying $\rep{R}_C = 0$: that is, models with strictly VL BSM particle content (as these evade the stringent constraints on electroweak-scale BSM fermions established by accelerator experiments). We refer to these as VL extensions.

To facilitate a quantitative description of our results, we introduce some useful notation. In particular, we define $N_{AB}, N_{noMCF}, N_{no\overline{SM}}, N_Y, N_U, N_{VL} : \mathcal{E} \to \mathbb{Z}_{\geq0}$ as follows. Let $E \in \mathcal{E}$ be given. We let $N_{AB}(E)$ denote the number of \text{ABs} resulting in decomposed representations $\rep{R}_x$ containing $\rep{R}_{SM}$ found by our algorithm for the NAB $E$. We let $N_{noMCF}(E)$ and $N_{no\overline{SM}}(E)$ denote the number of these successful ABs which yield a decomposed $\rep{R}^\prime$ (for $G_{SM}$) containing \textit{no massless charged fermions} and \textit{no BSM chiral particles dual to SM particles} (i.e., particles which would make any of the SM particles vector-like). That is, $N_{noMCF}(E)$ and $N_{no\overline{SM}}$ describe the ABs surviving our phenomenological constraints when imposed sequentially. We let $N_Y(E)$ denote the number of unique hypercharges featured in the ABs arising from $E$. I.e., we identify ABs which have arrive at the same \textit{hypercharge} through different representation-SM particle identifications. We let $N_U(E)$ denote $|\mathcal{M}_E|$, or the number of \textit{unique} representations associated to all surviving ABs for the NAB $E$, i.e., identifying ABs which yield the same particle content. Finally, we let $N_{VL}(E)$ denote the number of elements in $\mathcal{M}_E$ satisfying $\rep{R}_C = 0$, i.e., the number of unique \textit{vector-like} extensions associated to $E$.

Given these definitions, we now present a general summary of our results, with specifics regarding each $E \in \mathcal{E}$ such that $N_{noMCF}(E) > 0$ being characterized in Table I\ref{table:results}. Each $E \in \mathcal{E}$ does yield $N_{AB}(E) > 0$, and over the entire set $\mathcal{E}$ we find the following.
\begin{equation}
    \begin{aligned}
        \sum_{E \in \mathcal{E}} N_{AB}(E) &= \mathcal{O}(1.6 \times 10^6) \\
        \sum_{E \in \mathcal{E}} N_{noMCF}(E) &= \mathcal{O}(7 \times 10^4) \\
        \sum_{E \in \mathcal{E}} N_{no\overline{SM}} &= \mathcal{O}(3.6 \times 10^4) \\
        \sum_{E \in \mathcal{E}} N_Y &= \mathcal{O}(2.8 \times 10^3) \\
        \sum_{E \in \mathcal{E}} N_U(E) = \sum_{E \in \mathcal{E}} |\mathcal{M}_E| &= 151 \\
        \sum_{E \in \mathcal{E}} N_{VL}(E) &= 9
    \end{aligned}
\end{equation}
In the next section, we investigate the VL extensions, each of which appear to belong to infinite classes of VL extensions if we consider infinite classes of $abc$ product gauge groups with dimension extending beyond that of $E_6$. Following this, we consider the chiral extension models (where $\mathcal{R}_C \neq 0$) with the fewest BSM chiral particles, which potentially also remain phenomenologically viable.
\begin{table}[]
\begin{adjustwidth}{-1in}{-1in}
\tiny
\centering
\begin{tabular}{|c|c|c|c|c|c|c|c|c|c|}
    \hline
    Gauge Group ($G \in \mathcal{G}'$) & $\dim G$ & $N_P(G)$ & \begin{tabular}{c} NAB \\ ($E \in \mathcal{E}$)\end{tabular} & $N_{AB}(E) $ & $N_{noMCF}(E)$ & $N_{no\overline{SM}}$ & $N_Y(E)$ & $N_U(E)$ & $N_{VL}(E)$ \\
    \hline
    \hline
    $SU(4) \times SU(3) \times SU(3)$
    & $31$ & $108$ & $(4,3,3)$ & 84 & 84 & 48 & 9 & 2 & 1 \\
    \hline
    \hline
    $SU(4) \times SU(4) \times SU(3)$
    & $38$ & $144$ & $(3,4,4)$ & 696 & 444 & 252 & 24 & 1 & 0 \\
    \hline
    \hline
    $SU(5) \times SU(3) \times SU(3)$
    & $40$ & $135$ & $(5,3,3)$ & 1086 & 1086 & 516 & 57 & 7 & 1 \\
    \hline
    \hline
    $SU(5) \times SU(4) \times SU(4)$
    & $54$ & $240$ & $(5,4,4)$ & 20880 & 9148 & 5124 & 440 & 17 & 0 \\
    \hline
    \hline
    $SU(5) \times SU(5) \times SU(3)$
    & $56$ & $225$ & $(3,5,5)$ & 16020 & 1280 & 1074 & 120 & 3 & 0 \\
    \hline
    \hline
    \multirow{2}{*}{\centering $SU(6) \times SU(4) \times SU(3)$}
    & \multirow{2}{*}{$58$} & \multirow{2}{*}{$216$} & $(4,3,6)$ & 20520 & 2496 & 2304 & 67 & 2 & 0 \\
    \cline{4-10}
    & & & $(4,6,3)$ & 4572 & 252 & 252 & 48 & 1 & 1 \\
    \hline
    \hline
    $SU(5) \times SU(5) \times SU(4)$
    & $63$ & $300$ & $(4,5,5)$ & 48400 & 4910 & 4360 & 353 & 13 & 0 \\
    \hline
    \hline
    $SU(7) \times SU(3) \times SU(3)$ 
    & $64$ & $189$ & $(7,3,3)$ & 9870 & 9870 & 4920 & 537 & 55 & 5 \\
    \hline
    \hline
    \multirow{2}{*}{\centering $SU(6) \times SU(5) \times SU(3)$ }
    & \multirow{2}{*}{$67$} & \multirow{2}{*}{$270$} & $(5,3,6)$ & 74370 & 5024 & 3264 & 93 & 4 & 0 \\
    \cline{4-10}
    & & & $(5,6,3)$ & 14352 & 468 & 336 & 60 & 1 & 1  \\
    \hline
    \hline
    $SU(7) \times SU(4) \times SU(4)$
    & $78$ & $336$ & $(7,4,4)$ & 78696 & 35222 & 13940 & 1008 & 45 & 0 \\
    \hline
\end{tabular}
\label{table:results}
\caption{Using the definitions of Sec. 5, this table gives a summary of SM embeddings for all $E \in \mathcal{E}$ such that $N_{no\overline{SM}}(E) > 0$.}
\end{adjustwidth}
\end{table}

Regarding this table, it is worth noting that the $(N,3,3)$ NABs have the unique property that all successful ABs automatically satisfy the first phenomenological constraint: that is, the constraint does not act as a constraint at all for these NABs.  

In general, to specify the embedding of a model in the initial gauge group $SU(a)\times SU(b)\times SU(c)$ we need to write the hypercharge in terms of the $U(1)$ charges. When $SU(3)_C$ is in $SU(a)$ and $SU(2)_L$ is in $SU(b)$ we write the $U(1)$ charges as $A_1, A_2, \dots, A_{a-3}; B_1, B_2, \dots, B_{b-2}; C_1, C_2, \dots, C_{c-1}$. Given a model, by which we mean a $\rep{R} = \rep{R}_{SM} + \rep{R}_{VL} + \rep{R}_C$ for the group $G_{SM}$, the choice of hypercharge operator $Y$ yielding that model from among all possible linear combinations of these $A_i, B_j, C_k$ is hardly ever unique: thus, in each case that follows we merely present a prototypical form for $Y$.

\section{Example Models}

Here we consider examples of viable models found by our analysis of SBTs. First we consider the pristine vector-like extension of the SM and then the lowest lying chiral extension models. Those with the minimal BSM chiral content are less likely to give difficulty satisfying phenomenological constraints.

\subsection{Vector-Like Extensions}

\subsubsection{N33 Classes}

We conjecture the existence of five infinite classes of VL extensions arising from $G = SU(N) \times SU(3) \times SU(3)$ with the NAB $SU(N) \to SU(3)_C,\, SU(3) \to SU(2)_L$. The $i$th such VL extension model is denoted by $\rep{R}^{N33,i}$, where $N$ cannot be divisible by $3$ (otherwise the model will not support $3$ SM families). We require $N \geq 4$ for $i = 1$, and $N \geq 7$ for $1 < i \leq 5$. We find that
\begin{equation}
    \begin{aligned}
      \rep{R}^{N33,i} &=\rep{R_\text{SM}} +\rep{R}^{N33,i}_\text{universal}+\rep{R}^{N33,i}_\text{unique}
     \end{aligned}
\end{equation}
where
\begin{equation}
    \begin{aligned}
        \rep{R}^{N33,i}_\text{universal} &= 
         3(\rep{3},\rep{1})_{\frac{1}{6}} + 3({\repbar{3}},\rep{1})_{-\frac{1}{6}} \\
        &\qquad + N(\rep{1},\rep{2})_{\pm\frac{1}{2}} + (4N - 18)(\rep{1},\rep{2})_0 \\
        &\qquad + (4N - 15)(\rep{1},\rep{1})_{\pm\frac{1}{2}} + (7N - 36)(\rep{1},\rep{1})_0  
    \end{aligned}
\end{equation}
and where  $\rep{R}^{N33,i}_\text{unique}$ takes one of the 5   following forms
\begin{align*}
    \rep{R}^{N33,1}_\text{unique}(N) &= 6(\rep{1},\rep{2})_0 + 6(\rep{1},\rep{1})_{\pm\frac{1}{2}} + 12(\rep{1},\rep{1})_0 \\
    \rep{R}^{N33,2}_\text{unique}(N) &= 3(\rep{1},\rep{2})_{\pm\frac{1}{2}} + 3(\rep{1},\rep{1})_{\pm 1} + 6(\rep{1},\rep{1})_{\pm\frac{1}{2}} + 6(\rep{1},\rep{1})_0\\
    \rep{R}^{N33,3}_\text{unique}(N) &= 3(\rep{1},\rep{2})_{\pm1} + 3(\rep{1},\rep{1})_{\pm\frac{3}{2}}  + 6(\rep{1},\rep{1})_{\pm1} + 3(\rep{1},\rep{1})_{\pm\frac{1}{2}} \\
    \rep{R}^{N33,4}_\text{unique}(N) &= 3(\rep{1},\rep{2})_{\pm\frac{3}{2}} + 3(\rep{1},\rep{1})_{\pm2}  + 6(\rep{1},\rep{1})_{\pm\frac{3}{2}} + 3(\rep{1},\rep{1})_{\pm1} + 3(\rep{1},\rep{1})_{\pm\frac{1}{2}} \\
    \rep{R}^{N33,5}_\text{unique}(N) &= 3(\rep{1},\rep{2})_{\pm\frac{1}{10}} + 3(\rep{1},\rep{1})_{\pm\frac{3}{5}} + 3(\rep{1},\rep{1})_{\pm\frac{2}{5}} + 6(\rep{1},\rep{1})_{\pm\frac{1}{10}}
\end{align*}
The terms in $ \rep{R}^{N33,i}_\text{universal}$
are quarks and antiquarks of electric charge $\pm \frac{1}{6}$, lepton of charge $\pm \frac{1}{2}$ and normally charged leptons.
Hence, this leads to heavy fractionally charged color singlets, and multiply charged magnetic monopoles.

Particles in the $ \rep{R}^{N33,i}_\text{unique}(N)$ representations are all leptons. Note all 5 $ \rep{R}^{N33,i}_\text{unique}(N)$ cases contain 36 states. For $i = 1,2$ the charges are all standard except of the electric charge $\pm \frac{1}{2}$ singlets. In $i = 3, 4$, we have leptons with exotic charges, namely $\pm 2$, $\pm \frac{3}{2}$, $\pm 2$, and $\pm \frac{5}{2}$ as well as normally charged leptons. Finally, $i = 5$ is the most exotic case with only non SM charges ranging from $\pm \frac{1}{10}$ to $\pm \frac{11}{10}$. Considering the trend from $i = 1$ to $i = 4$, it seems possible that sufficiently large $N$ might yield new classes with particles of the form $(\rep{1},\rep{2})_{\pm \frac{n}{2}}, (\rep{1},\rep{1})_{\pm\frac{m}{2}}$ for $n > 3, m > 4$.

Any of these models with heavy fractionally charged color singlets, must have a lightest such particle.
Since there is nothing lighter in the SM for which it to decay into, it must be stable, similar to a lightest supersymmetric partner (LSP) found in SUSY models.

Let $Y^{N33,i}(N)$ denote a $U(1)_Y$ charge operator associated with the symmetry breaking that leads  to $\rep{R}^{N33,i}$. Example of each case are
\begin{align*}
    Y^{N33,1}(N) &= \frac{1}{6} A_{N-3} + \frac{1}{4} C_1 + \frac{1}{4} C_2\\
    Y^{N33,2}(N) &= \frac{1}{10} A_{N-5} - \frac{1}{10} A_{N-4} + \frac{1}{6} A_{N-3} + \frac{1}{2} C_2\\
    Y^{N33,3}(N) &= \frac{1}{5} A_{N-5} - \frac{1}{5} A_{N-4} + \frac{1}{6} A_{N-3} + \frac{1}{2} C_2\\
    Y^{N33,4}(N) &= \frac{3}{10} A_{N-5} - \frac{3}{10} A_{N-4} + \frac{1}{6} A_{N-3} + \frac{1}{2} C_2\\
    Y^{N33,5}(N) &= \frac{1}{10} A_{N-5} + \frac{1}{15} A_{N-3} + \frac{1}{2} C_2
\end{align*}
We have confirmed these $5$ classes through $N = 11$.\\

\subsubsection{N63 Classes}

We conjecture the existence of two infinite classes of VL extensions arising $G = SU(N) \times SU(6) \times SU(3)$ with the NB $SU(N) \to SU(3)_C,\, SU(6) \to SU(2)_L$. The $i$th such VL extension model is denoted by $\rep{R}^{N63,i}$, where $N$ cannot be divisible by $3$, $N \geq 4$ for $i = 1$, and $N \geq 7$ for $i = 2$. We note that even the minimal case $N = 7$ in this second class of models is not an SBT (i.e., it has $91$ gauge generators) but we include the class here nevertheless for completeness. We find that
\begin{equation}
    \begin{aligned}
      \rep{R}^{N63,i} &=\rep{R_\text{SM}} +\rep{R}^{N63,i}_\text{universal}+\rep{R}^{N63,i}_\text{unique}
     \end{aligned}
\end{equation}
where
\begin{equation}
    \begin{aligned}
        \rep{R}^{N63,i}_\text{universal} 
        &= 3(\rep{3},\rep{1})_{\frac{2}{3}} + 3(\repbar{3},\rep{1})_{-\frac{2}{3}} + 3(\rep{3},\rep{1})_{-\frac{1}{3}} + 3(\repbar{3},\rep{1})_{\frac{1}{3}} + 6(\rep{3},\rep{1})_{\frac{1}{6}} + 6(\repbar{3},\rep{1})_{-\frac{1}{6}} \\
        &\qquad + N(\rep{1},\rep{2})_{\pm\frac{1}{2}} + (4N-18)(\rep{1},\rep{2})_{0} \\
        &\qquad + (N+3)(\rep{1},\rep{1})_{\pm1} + (12N - 48)(\rep{1},\rep{1})_{\pm\frac{1}{2}} + (16N - 42)(\rep{1},\rep{1})_{0} 
   \end{aligned}
\end{equation}
and where the $\rep{R}^{N63,i}_\text{unique}$ are
\begin{align*}
    \rep{R}^{N63,1}_\text{unique}(N) &= 6(\rep{1},\rep{2})_{0} + 18(\rep{1},\rep{1})_{\pm\frac{1}{2}} \\
    \rep{R}^{N63,2}_\text{unique}(N) &= 3(\rep{1},\rep{2})_{\pm\frac{1}{10}} + 9(\rep{1},\rep{1})_{\pm\frac{3}{5}} + 9(\rep{1},\rep{1})_{\pm\frac{2}{5}} + 12(\rep{1},\rep{1})_{\pm\frac{1}{10}} 
\end{align*}
Let $Y^{N63,i}(N)$ denote the $U(1)_Y$ hypercharge operator associated with the symmetry breaking procedure corresponding to $\rep{R}^{N63,i}$. Example forms for each of these cases are
\begin{align*}
    Y^{N63,1}(N) &= \frac{1}{6}A_{N-3} +  \frac{1}{6}B_4 -  \frac{1}{6}B_3 + \frac{1}{2}C_2\\
    Y^{N63,2}(N) &= \frac{1}{10}A_{N-5} + \frac{1}{15}A_{N-3} +  \frac{1}{6}B_4 -  \frac{1}{6}B_3 + \frac{1}{2}C_2
\end{align*}
We have also confirmed these $2$ classes through $N = 11$.\\

\subsection{Minimal Chiral Extensions}

Any fermion content beyond the three families of the SM is difficult to accommodate given current accelerator physics constraints \cite{PDG}. However, there are a number of low  lying chiral  extension of the SM that arise in our analysis that we include here for completeness as they may not all be totally ruled out.

If we identify models (of which we have $151$) with identical chiral BSM particle content, we are left with $123$ equivalence classes. In this section, we briefly consider equivalence classes with the fewest number of chiral particles beyond the SM, some of which---while they must have BSM fermions  masses near the electroweak scale---remain possibly phenomenologically viable. In particular, we present $\rep{R}_C$ for all models with up to $30$ BSM chiral particles, of which there are $13$ equivalence classes (corresponding to $38$ unique particle configurations if one distinguishes models with the same chiral content and distinct vector-like content). One of these equivalence classes, containing $9$ unique particle configurations, is where $\rep{R}_C$ is empty: the set of VL extensions, which we've already discussed. Thus, in this section we discuss the remaining $12$ equivalence classes ($29$ unique particle configurations). In each case, the NAB often is not unique, while the AB never is. Thus, we merely present one particular NAB-AB pair that achieves the model.
 
First we have $\rep{R}_C^{11}$ with $11$ (leptonic) particles as follows, where  superscripts will always denote the number of BSM chiral particles, with a letter appended to distinguish between distinct models with the same particle numbers.  We find
\begin{equation}
    \rep{R}_C^{11} = (\rep{1},\rep{2})_{\frac{3}{2}}+3(\rep{1},\rep{2})_{-\frac{1}{2}}+(\rep{1},\rep{1})_{-2}+2(\rep{1},\rep{1})_{1}
\end{equation}
This model can be realized only through the NAB $(7,3,3)$, for which we can choose $Y$ as 
\begin{equation}
    Y^{11} = \frac{1}{4}A_3 + \frac{5}{12}A_4 + \frac{1}{2}B_1 + C_2.
\end{equation}
Most particles in eq. (18) can acquire mass from the SM Higgs doublet, e.g., 
the mass terms for two of the three standard hypercharged doublets and right-handed singlets are just as in the SM. The $Y = \frac{3}{2}$ doublet mass term  of the form
\begin{equation}
    \begin{aligned}
        (\mathbf{1},\mathbf{2})^F_{\frac{3}{2}}(\mathbf{1},\mathbf{1})^F_{-2}(\mathbf{1},\mathbf{2})^H_{-\frac{1}{2}}
    \end{aligned}
\end{equation}
 generates a Dirac electric charged $2$ massive fermion.   We can write a further mass term using the third standard hypercharged doublet and the $(\mathbf{1},\mathbf{2})^F_{\frac{3}{2}}$.
An example of such a term could involve a Higgs triplet and be
\begin{equation}
    \begin{aligned}
        (\mathbf{1},\mathbf{2})^F_{\frac{3}{2}}(\mathbf{1},\mathbf{2})^F_{-\frac{1}{2}}(\mathbf{1},\mathbf{3})^H_{-1}
    \end{aligned}
\end{equation}
Hence we must extend the Higgs sector to avoid massless charged fermions.
Including dimension 5 operators could also accomplish the purpose, but would typically lead to a light charged lepton the violates observational bounds. Hence dimension 4 mass term are preferable. 

Next, we have $\rep{R}_C^{15a}$ with extra (leptonic) chiral particles 
\begin{equation}
    \rep{R}_C^{15a} = (\rep{1},\rep{2})_{-\frac{1}{2}} + 3(\rep{1},\rep{2})_{\frac{1}{6}} + (\rep{1},\rep{1})_1 + 3(\rep{1},\rep{1})_{-\frac{2}{3}} + 3(\rep{1},\rep{1})_{\frac{1}{3}}.
\end{equation}
This model can be realized through NABs $(N,3,3)$ for $N = 4,7$. For $(4,3,3)$, we can choose $Y$ as follows. 
\begin{equation}
    Y^{15a} = -\frac{1}{6}B_1 + \frac{1}{3}C_1
\end{equation}
The scalar sector can be simpler this time, since all mass terms require only standard Higgs doublets to deliver the extra SM-like term 
\begin{equation}
    \begin{aligned}
(\mathbf{1},\mathbf{2})^F_{-\frac{1}{2}}(\mathbf{1},\mathbf{1})^F_{1}(\mathbf{1},\mathbf{2})^H_{\frac{1}{2}}
\end{aligned}
\end{equation}
and three of the form
\begin{equation}
    \begin{aligned}
        (\mathbf{1},\mathbf{2})^F_{\frac{1}{6}}[(\mathbf{1},\mathbf{1})^F_{-\frac{2}{3}}+(\mathbf{1},\mathbf{1})^F_{\frac{1}{3}}](\mathbf{1},\mathbf{2})^H_{\frac{1}{2}}
    \end{aligned}
\end{equation}
Hence we have an extra SM electric charge 1 lepton plus three charge $\frac{2}{3}$ and three charge $\frac{1}{3}$ leptons, all with masses not far above the electroweak scale.
We also have triply charged magnetic monopoles in this model.

We next have $\rep{R}_C^{15b} = \rep{F}$, i.e., a fourth SM family. We exclude the extra left-handed antineutrino in our count because it can be VL with itself, and we do this every time we refer to $\rep{R}_{SM}$ in this section. This model can be realized through NABs $(N,4,4)$ for $N = 5,7$. For $(5,4,4)$, we can choose $Y^{15b}$ as follows.
\begin{equation}
    Y^{15b} = \frac{1}{6}A_2 + \frac{1}{4}C_2 + \frac{1}{4}C_3
\end{equation}

Next $\rep{R}_C^{16}$ contains only leptons
\begin{equation}
    \rep{R}_C^{16} = (\rep{1},\rep{2})_{\frac{9}{2}}+3(\rep{1},\rep{2})_{-\frac{3}{2}}+(\rep{1},\rep{1})_{-5}+(\rep{1},\rep{1})_{-4}+3(\rep{1},\rep{1})_{2}+3(\rep{1},\rep{1})_{1}
\end{equation}
But exotic electromagnetic charges like $Q=5$ would have distinctive signatures in HEP experiments.
This model can be realized only through the NAB $(7,3,3)$, for which we can choose $Y^{16}$ as follows.
\begin{equation}
    Y^{16} = \frac{3}{5}A_2 + \frac{9}{10}A_3 + \frac{1}{6}A_4 + \frac{3}{2}B_1 - C_1 + 2C_2
\end{equation}

Continuing, we have $\rep{R}_C^{24}$, which  contains a forth family of quarks, but only exotic leptons in the extended sector.
\begin{equation}
    \rep{R}_C^{24} = (\rep{3}, \rep{2})_{\frac{1}{6}} + (\repbar{3}, \rep{1})_{-\frac{2}{3}} + (\repbar{3}, \rep{1})_{\frac{1}{3}} + 
    3(\rep{1}, \rep{2})_{-\frac{1}{6}} + 3(\rep{1}, \rep{1})_{\frac{2}{3}} + 3(\rep{1}, \rep{1})_{-\frac{1}{3}}.
\end{equation}
This model can be realized through the NABs $(N,4,4)$ for $N = 3,7$; for $(3,4,4)$ we can choose $Y^{24}$ as follows.
\begin{equation}
    Y^{24} = -\frac{1}{6}B_2 - \frac{1}{3}C_2
\end{equation}

Next we have the leptonic $\rep{R}_C^{26}$ where
\begin{equation}
    \rep{R}_C^{26} = 3(\rep{1},\rep{2})_{\frac{5}{2}} + 5(\rep{1},\rep{2})_{-\frac{3}{2}} + 3(\rep{1},\rep{1})_{-3} + 2(\rep{1},\rep{1})_{2} + 5(\rep{1},\rep{1})_{1}.
\end{equation}
This model can be achieved only through the NAB $(5,3,3)$, for which we can choose $Y^{26}$ as 
\begin{equation}
    Y^{26} = -\frac{1}{2}A_1 + \frac{1}{6}A_2 - \frac{1}{2}B_1 + C_2.
\end{equation}

Another model with purely leptonic content in the extended chiral sector is $\rep{R}_C^{27}$.
\begin{equation}
    \rep{R}_C^{27} = 3(\rep{1},\rep{2})_{\frac{5}{6}} + 5(\rep{1},\rep{2})_{-\frac{1}{2}} + 3(\rep{1},\rep{1})_{-\frac{4}{3}} + 5(\rep{1},\rep{1})_{1} + 3(\rep{1},\rep{1})_{-\frac{1}{3}}
\end{equation}
This model can only be achieved through the NAB of $(5,3,3)$, for which we can choose $Y^{27}$ as  
\begin{equation}
    Y^{27} = \frac{1}{6}A_1 - \frac{1}{6}A_2 - \frac{1}{6}B_1 + \frac{1}{3}C_1
\end{equation}

Moving on, we have $\rep{R}_C^{29}$ with (leptonic) particles in the representations
\begin{equation}
    \rep{R}_C^{29} = 3(\rep{1},\rep{2})_{-\frac{1}{2}} + 5(\rep{1},\rep{2})_{\frac{3}{10}} + 3(\rep{1},\rep{1})_{1} + 5(\rep{1},\rep{1})_{-\frac{4}{5}} + 5(\rep{1},\rep{1})_{\frac{1}{5}}
\end{equation}
This model can only be achieved through the NAB of $(5,3,3)$, for which we can choose $Y^{29a}$ as  
\begin{equation}
    Y^{29} = \frac{1}{10}A_1 + \frac{1}{6}A_2 + \frac{1}{10}B_1 + \frac{1}{10}C_1 + \frac{1}{2}C_2
\end{equation}

Finally we arrive at cases with 30 extra fermions. We have $\rep{R}_C^{30a} = 2\rep{R}_C^{15a}$, achievable only in the NAB of $(5,3,3)$ with a choice of $Y^{30a}$ as 
\begin{equation}
    Y^{30a} = -\frac{1}{6}B_1 + \frac{1}{3}C_1
\end{equation}
Next we have $\rep{R}_C^{30b}$ with (leptonic) particles
\begin{equation}
    \rep{R}_C^{30b} = 3(\rep{1},\rep{2})_{\frac{11}{6}} + 3(\rep{1},\rep{2})_{-\frac{3}{2}} + 2(\rep{1},\rep{2})_{-\frac{1}{2}} + 3(\rep{1},\rep{1})_{-\frac{7}{3}} + 3(\rep{1},\rep{1})_{2} + 3(\rep{1},\rep{1})_{-\frac{4}{3}} + 5(\rep{1},\rep{1})_{1}.
\end{equation}
This model can only be achieved through the NAB of $(5,3,3)$, for which we can choose $Y^{30b}$ as follows.
\begin{equation}
    Y^{30b} = \frac{5}{12}A_1 - \frac{5}{12}A_2 - \frac{1}{6}B_1 + \frac{1}{3}C_1
\end{equation}
Next, we have $\rep{R}_C^{30c}$ with the following particles in the chiral extension
\begin{equation}
    \rep{R}_C^{30c} = \rep{F} + (\rep{1}, \rep{2})_{-\frac{1}{2}} + 3(\repbar{1}, \rep{2})_{\frac{1}{6}} + (\repbar{1}, \rep{1})_{1} + 3(\repbar{1}, \rep{1})_{-\frac{2}{3}} + 3(\repbar{1}, \rep{1})_{\frac{1}{3}}
\end{equation}
This model can be achieved through the NABs of $(N,4,4)$ for $N = 5,7$; for $(5,4,4)$, we can choose $Y^{30c}$ as   
\begin{equation}
    Y^{30c} = -\frac{1}{6}B_2 + \frac{1}{3}C_2
\end{equation}
Lastly, we have the possibility $\rep{R}_C^{30d} = 2\rep{F}$ realizable only through the NAB of $(4,5,5)$, with a choice of $Y^{30d}$ as 
\begin{equation}
    Y^{30d} = \frac{1}{6}A_1 + \frac{1}{2}C_3.
\end{equation}

\section{Summary and Conclusion}

\begin{figure}[!h]
    \centering
    \includegraphics[scale=0.3]{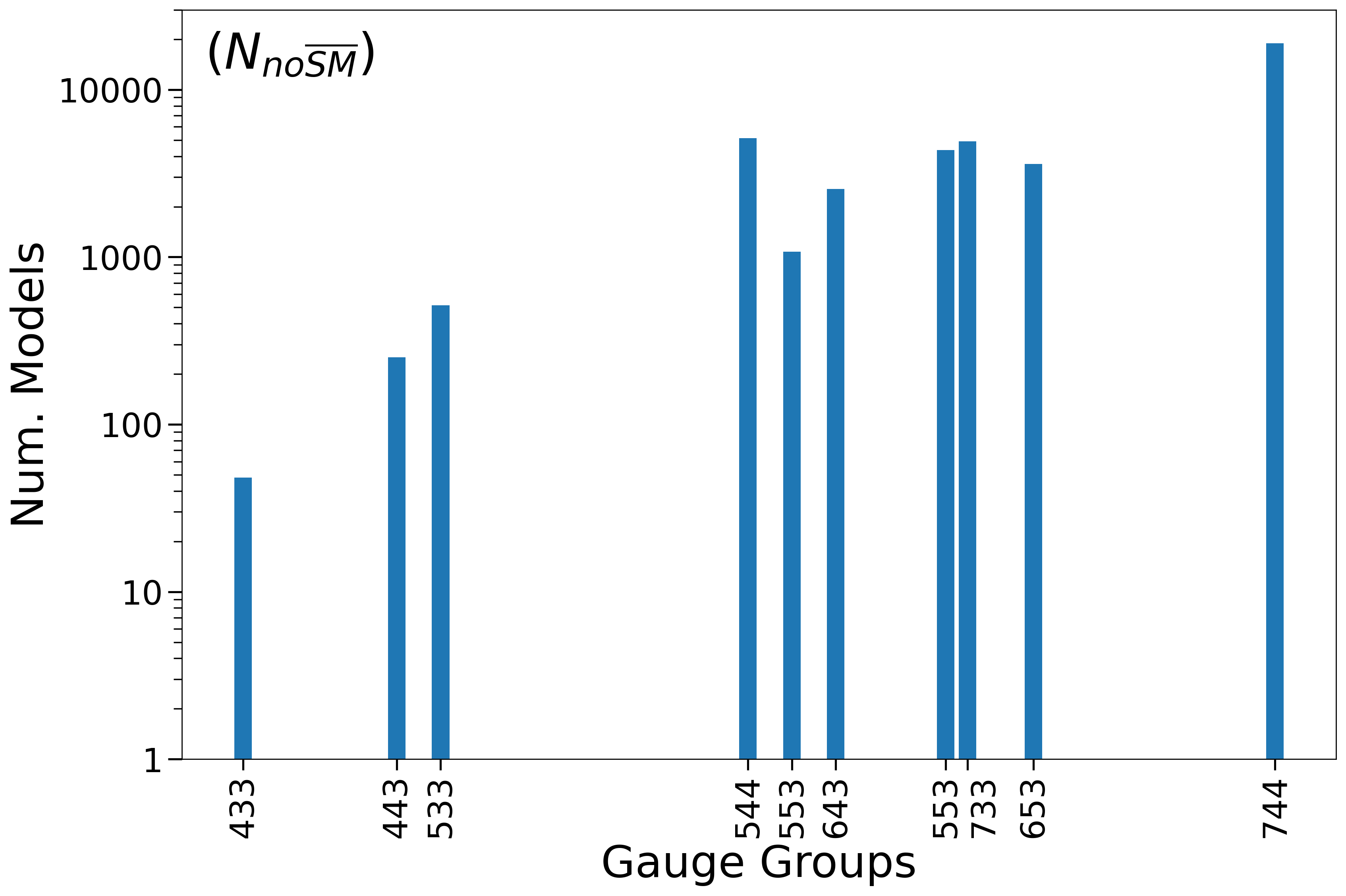}\hfill
    \includegraphics[scale=0.3]{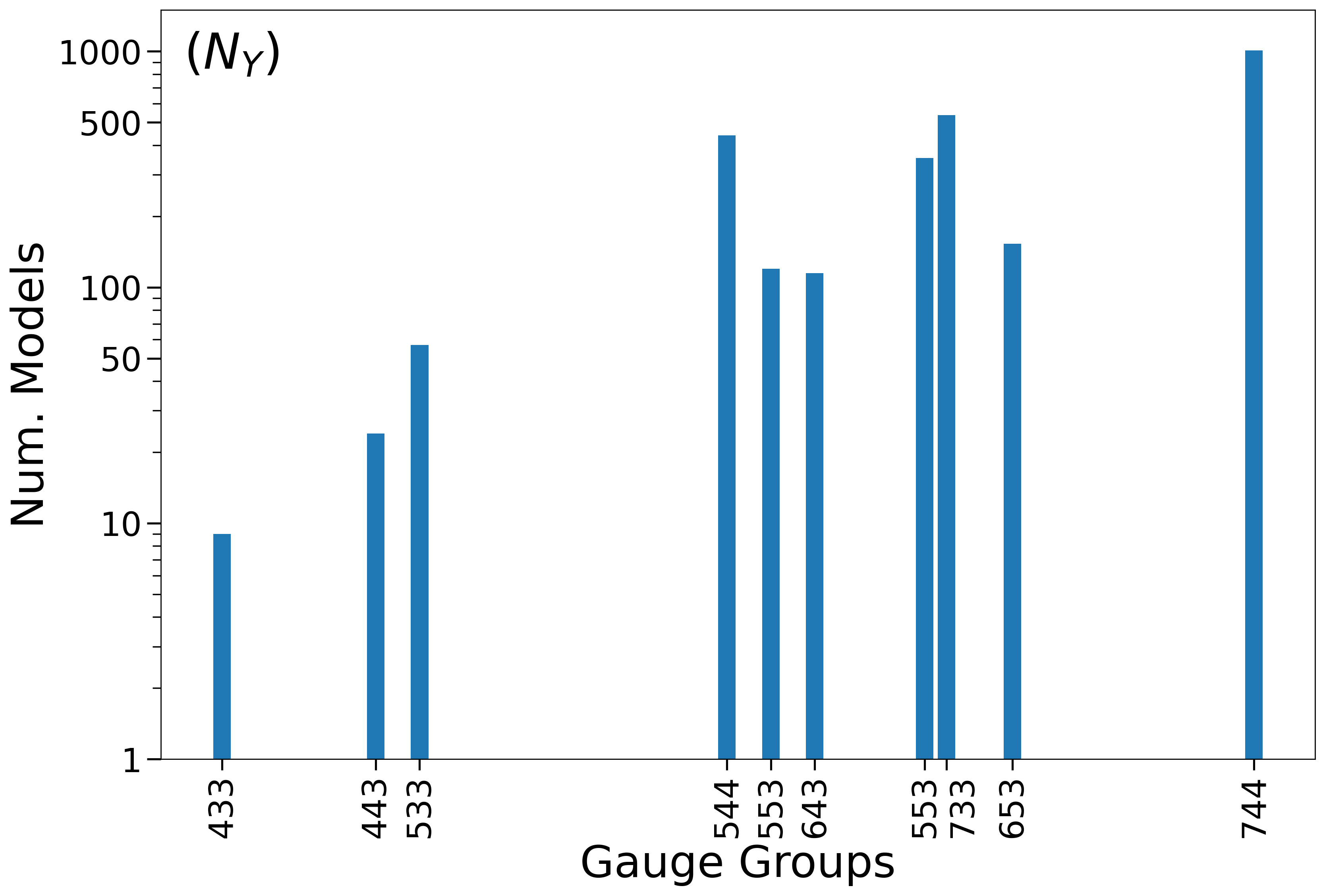}\hfill
    \includegraphics[scale=0.3]{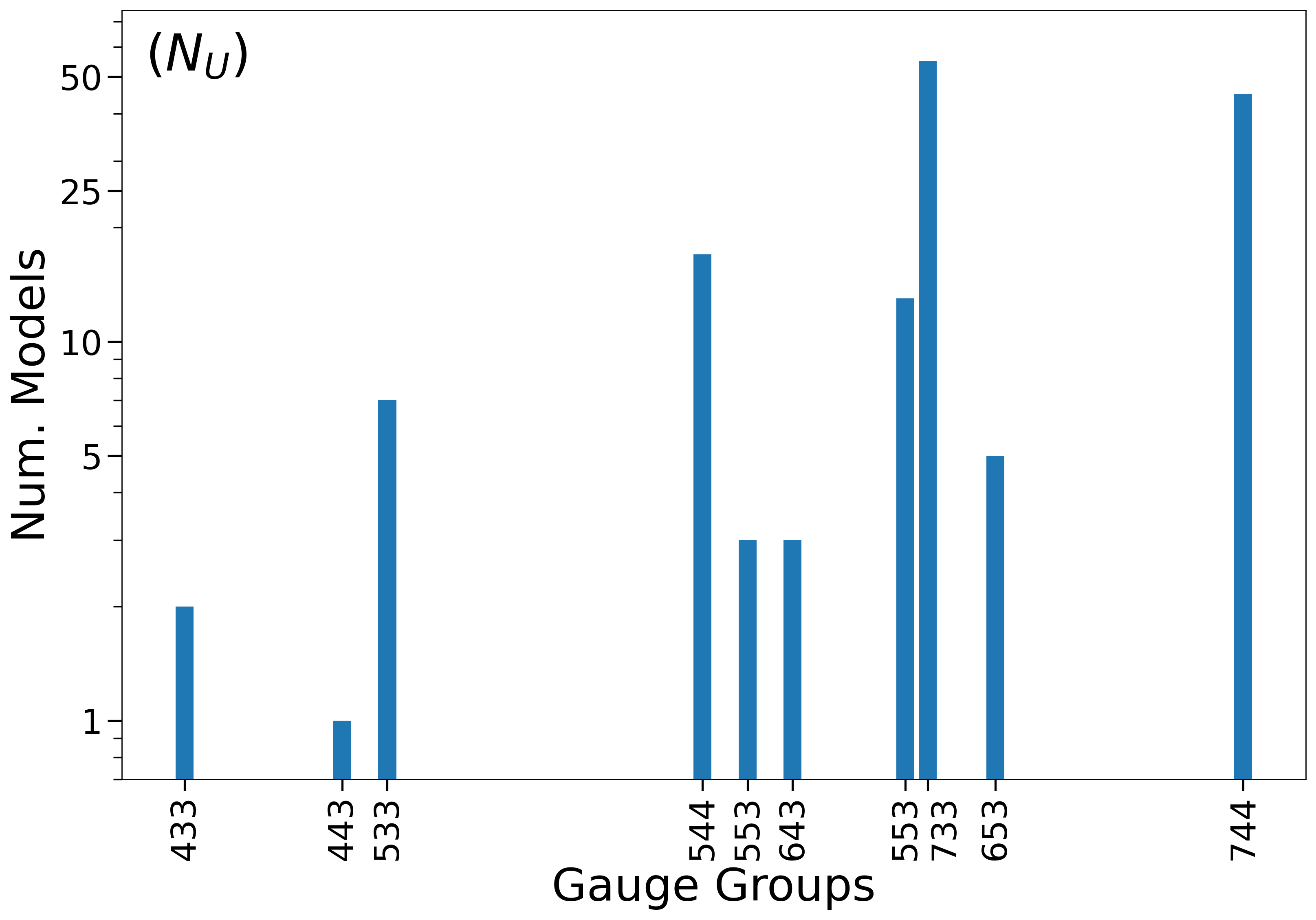}
    \caption{Summary of viable models  presented as plots, corresponding to each of our three methods of counting models, of the number of models we find as a function of the gauge group. }
    \label{summary_plots_gauge_groups}
\end{figure}

\begin{figure}[!h]
    \centering
    \includegraphics[scale=0.3]{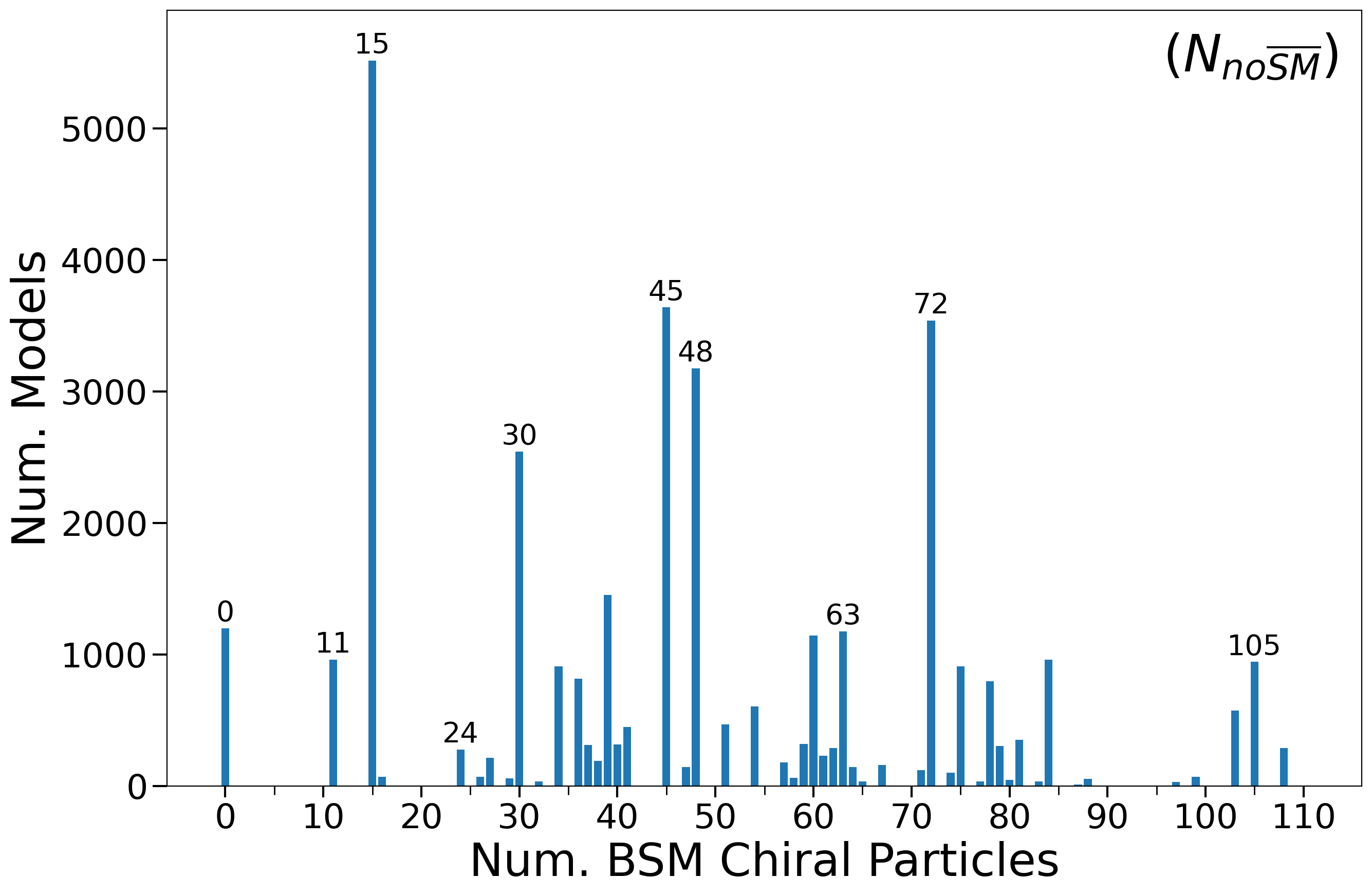}\hfill
    \includegraphics[scale=0.3]{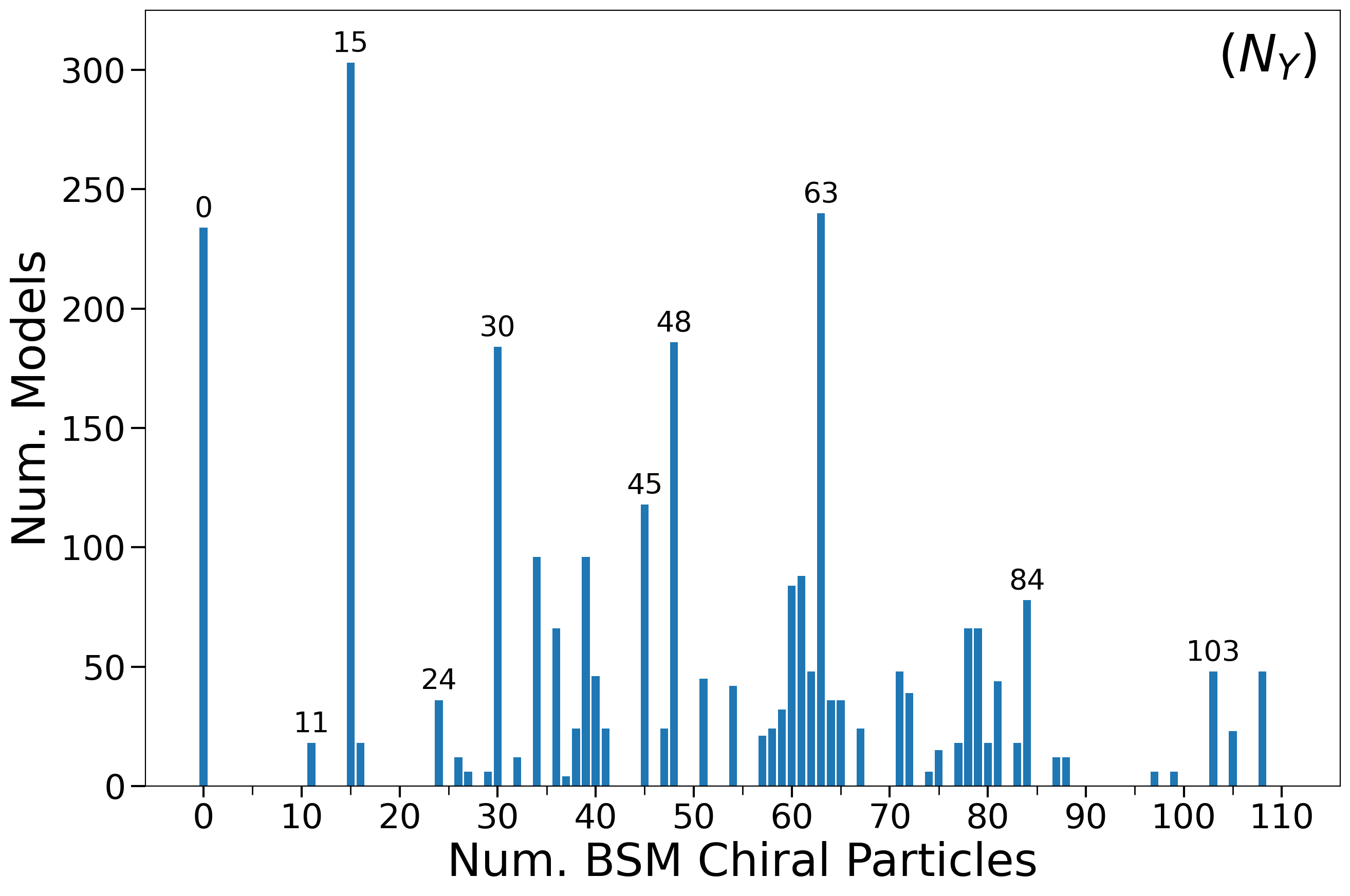}\hfill
    \includegraphics[scale=0.3]{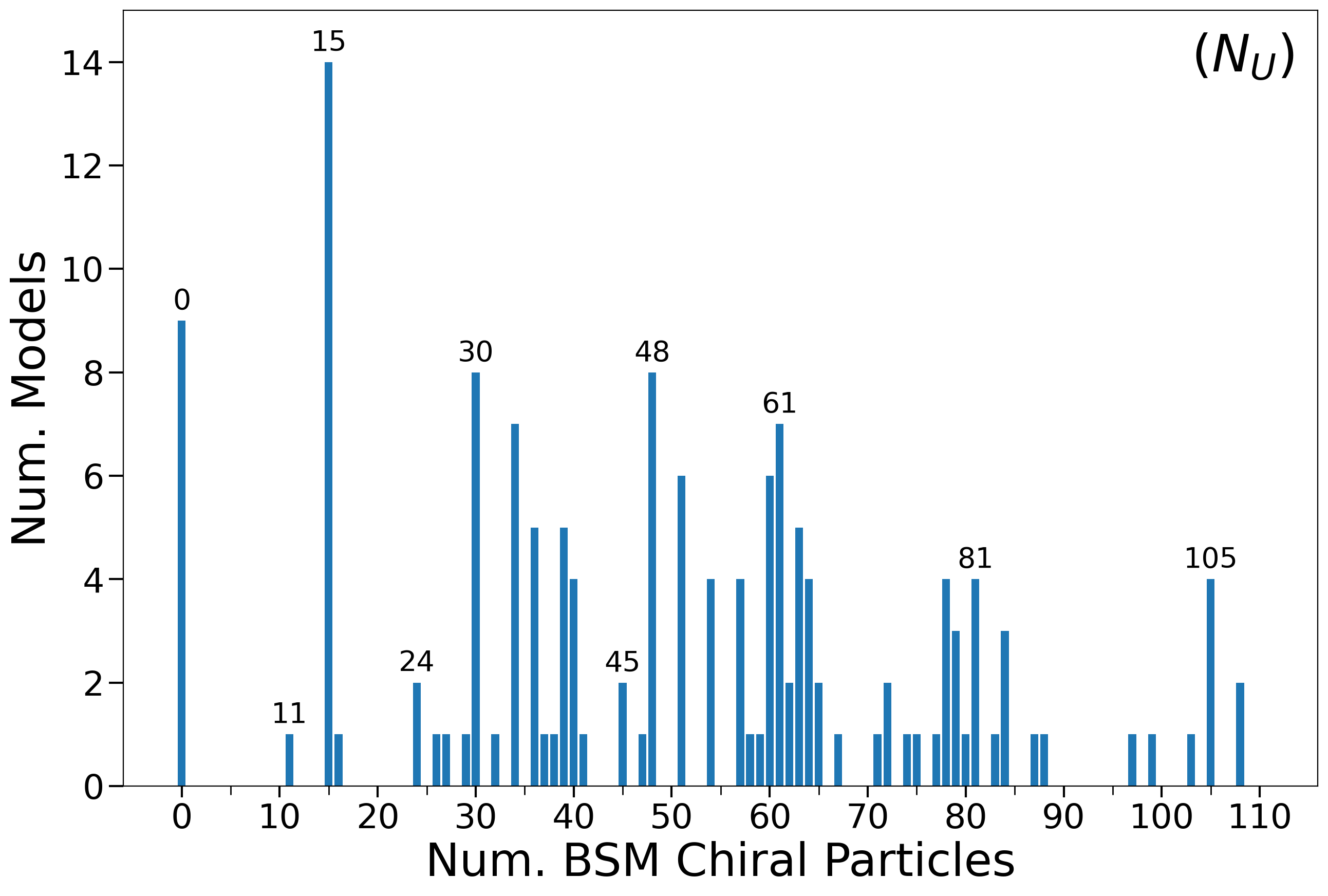}
    \caption{Summary of viable models 
presented as plots, corresponding to each of our three methods of counting models, of the distribution of the number of chiral BSM particles we find in all models across all NABs. }
    \label{summary_plots_all_NABs}
\end{figure}

\begin{figure}
    \includegraphics[scale=0.3]{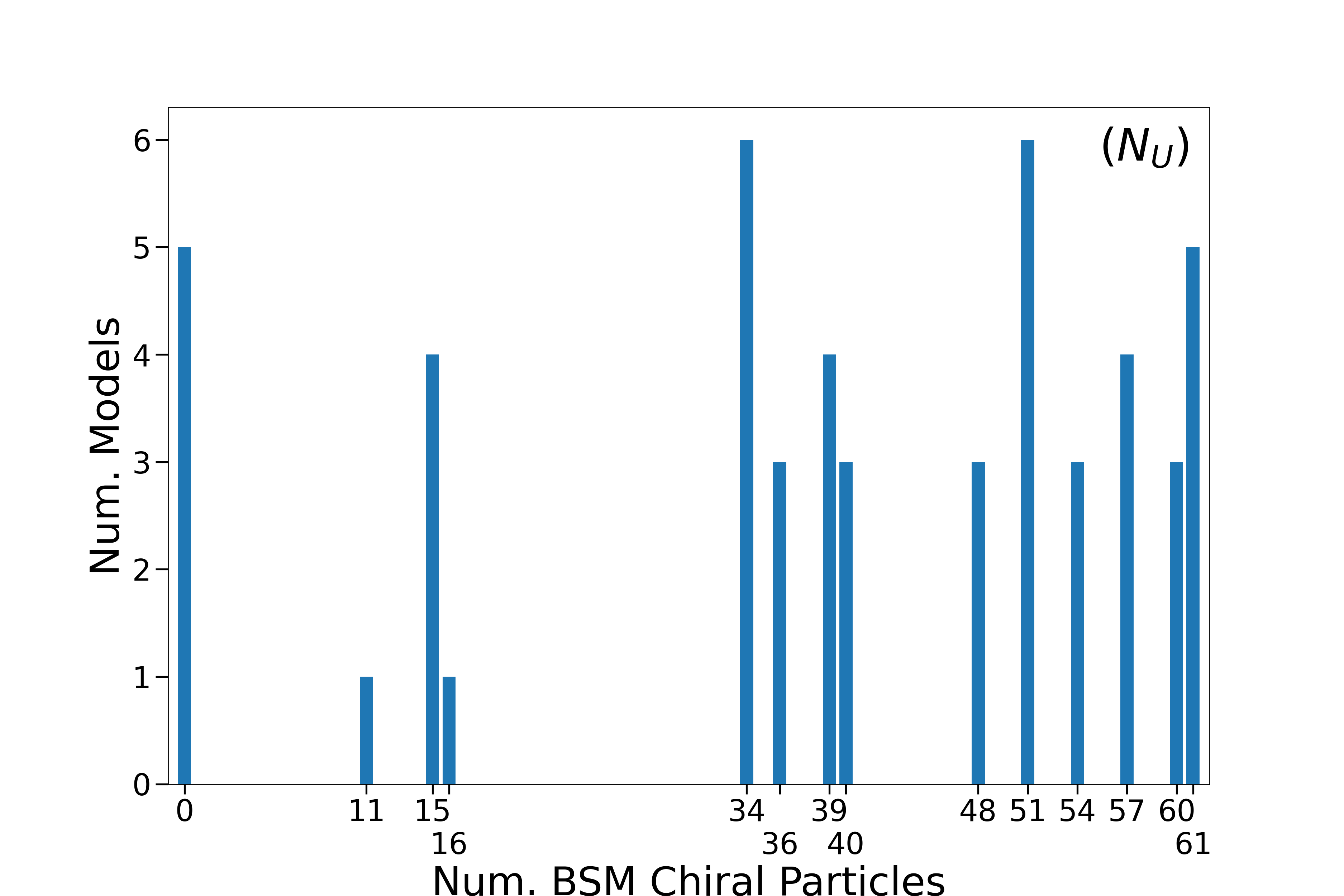}
    \caption{Summary of potentially viable $SU(7) \times SU(3) \times SU(3)$ models. }
    \label{733_plot}
\end{figure}
 
The results presented here are treated without mention of supersymmetry. However,  all models could
also potentially be ${\cal N}=1$ SUSY models assuming for instance they arise as orbifolded AdS/CFTs \cite{Lawrence:1998ja}.
 
Among the possible SBTs, we find that only a very small percentage of models have three SM families and nothing more. Of the $151$ three family models with no more than 78 generators, only 9 do not have extra chiral content. However, we do find that the models with strictly VL particle content beyond the SM are constituents of seemingly infinite classes of models with arbitrarily large gauge groups: in particular, we find seven such classes.

Table I efficiently summarizes  our findings: at this point, now we elaborate on those results through plots characterizing the distribution of number of models among our gauge groups and the number of chiral particles among our models.

However, we do have a choice to make: namely, what do we consider to be a model for the purposes of these plots? In Table I, we count models in a few different ways, even if we require that a decomposed representation $\rep{R}_x$ contains $\rep{R}_{SM}$ and satisfies our two phenomenological constraints. In particular, we could let a model be a successful AB (characterized by a choice of $w$ and $x$ as defined in our Search Algorithm section), which is how we counted models with $N_{no\overline{SM}}$; alternatively, when counting models we could identify ABs with the same $x$ (the same hypercharge), which corresponds to how we counted models for $N_Y$; or a model could mean a unique set of particles (i.e., identify ABs that yield the same BSM particles), which is how we counted models for $N_U$.

To be comprehensive, we have decided to provide the histograms associated with all three choices described in the above paragraph, and we always clarify which counting choice we have made in either the upper right or upper left corner of the plot by writing either $N_{no\overline{SM}}$, $N_Y$, or $N_U$. Now, having discussed the ways in which we count our models, we can  turn our attention to the figures themselves. 

Fig. \ref{summary_plots_gauge_groups} plots the number of models we find as a function of the gauge group. Note that we encounter two instances where a gauge group admits two viable NABs---$SU(6) \times SU(4) \times SU(3)$ and $SU(6) \times SU(5) \times SU(3)$---in which case we of course attribute both NABs to the given gauge group.     It is no  surprise that the number of models increases considerably with the dimension of the gauge group (note that the y-axis is on a log scale), although this increase is not strictly monotonic. We present plots corresponding to each of our three methods of counting models, and there are differences in shape between the three plots. For instance, we find that $SU(4) \times SU(4) \times SU(3)$ and $SU(5) \times SU(3) \times SU(3)$ yield similar numbers of viable numbers of ABs (as seen in the topmost plot) yet for the former group a considerably larger portion of these ABs end up yielding just a single unique configuration of particles, while the latter group's ABs culminate in $9$ models with unique particle content.

Fig. \ref{summary_plots_all_NABs} visualizes the distribution of number of chiral BSM particles we find in all models across all NABs. Again, we present plots corresponding to all three of our model counting methods. The bins with particularly large contents with respect to adjacent bins are labelled above the bins by their x-value in an ad hoc fashion to aid readability. In particular, at $0$ we find our vector-like extensions. (We can see $9$ of them in the bottommost plot, as we expect). We find that models with $15$, $30$, and $45$ chiral particles are common: we understand this to be the case because a SM family contains $15$ particles, and models with $1$, $2$, or $3$ extra families appear frequently.

Fig. \ref{733_plot} zooms in a little bit from the previous figure by taking a particular example, $(7,3,3)$, and exhibits the distribution of number of chiral BSM particles we find in each model. In particular, we see the $5$ vector-like extensions coming from $(7,3,3) \in \mathcal{E}$ appearing at $0$ chiral particles, our minimal chiral model $\rep{R}_{SM} + \rep{R}^{11}_C + \rep{R}_{VL}$ at $11$ chiral particles, a few instances of $15$ particle models, and so on. We note that the majority of models do indeed feature many chiral BSM particles, making them phenomenologically tenuous, as fermions with electroweak-scale masses are highly constrained.
 
Of those with extra chiral content, most have many extra particles that should have masses  near the electroweak scale.
The minimum number of extra chiral particles is 11, by the typical number is well over 50. In most of these models that means low mass exotic particles like fractionally charged hadrons and lepton, as well as multiply charged magnetic monopoles
are predicted. Even in the rare case where all the extra chiral particles have SM charges, they are still predicted to 
show up at low mass. Hence, making a viable phenomenology out of such model would be challenging to say the least. 
However, as they would definitely lead to new physics near the electroweak scale, those with the least extra chiral content beyond the three SM  families would be worthy of future study.
Perhaps an apt comparison is the low scale  ${\cal N}=1$ SUSY version of the SM, where the extra particle content creates phenomenological challenges. 

 As we mentioned in the introduction, our interest here has been in surveying  potentially viable quiver models with three family fermion spectra. Due to the vast number of possibilities, we have not focused attention on any particular models. However, we hope to study some of the phenomenological implications of the models presented here in future work and to put them in context with the literature.

\section{Appendix}
Here we prove that the minimal norm always generates particles whose hypercharges and hence their electric charges are rational fractions
of the electron charge, and as a result, charge quantization is preserved by the minimal norm choice.
\begin{theorem}
    \label{th}
    If a linear system $Ax = b$ ($A \in \mathbb{M}_{m, n}, b \in \mathbb{R}^m$) with rational coefficients has a solution, then the solution with least Euclidean norm is rational.
\end{theorem}
\begin{proof}
    By hypothesis, we may let $p \in \mathbb{R}^n$ be a solution: then the space of solutions is $S = \ker(A) + p$. Consider the set $S \cap \ker(A)^\perp$: we claim that this \textit{i)} contains one element $y$ which \textit{ii)} is rational and \textit{iii)} is the least norm solution. 
    
    To see \textit{i)}, let $Ax = b$ row-reduce to $\Tilde{A}x = \Tilde{b}$ and let $v_1, \dots, v_\ell$ be a basis for $\ker(A)$. If $y \in S \cap \ker(A)^\perp$\footnote{We recall that the kernel of a map between vector spaces is the subset of the domain mapped to $0$ in the codomain, and we use $V^\perp$ to denote the orthogonal complement of a linear subspace $V$, or the subspace of vectors perpendicular to each vector $v \in V$.}, then $y$ satisfies $A'y = b'$ where $A'$ is achieved by taking the $k$ non-zero rows of $\Tilde{A}$ and appending $v_1, \dots, v_\ell$ as rows while $b'$ is the non-zero entries of $\Tilde{b}$ with zeros appended 
    $A'$ is an element of $\mathbb{M}_n$ (by the rank-nullity theorem) and is full rank (by the linear independence of row and null spaces), thus the solution $y$ exists and is unique.
    
    Then \textit{ii)} follows readily: if $A, b$ are rational then $A', b'$ are as well and the rationality of $y$ then follows from Cramer's rule.
    
    Finally, for \textit{iii)}, note that $x \in S$ decomposes uniquely as $y + z$ for some $z \in \ker(A)$; thus, $\|x\|^2 = \|y\|^2 + \|z\|^2$ by the orthogonality of $y$ and $z$, so $\|x\|$ is minimized for $z = 0$ or $x = y$.
\end{proof}
\vspace{.5in}
\noindent {\it Acknowledgments:} This work was supported by US DOE grant DE-SC0019235. ES would like to thank Connor Lehmacher for a useful conversation regarding Th. \ref{th}. TWK thanks the Aspen Center for Physics, which is supported by National Science Foundation grant PHY-1607611, for hospitality while this paper was being completed.

\end{document}